\documentclass[12pt,reqno]{article}

\usepackage[usenames]{color}
\usepackage{amssymb}
\usepackage{graphicx}
\usepackage{pstricks}

\usepackage[colorlinks=true,
linkcolor=webgreen,
filecolor=webbrown,
citecolor=webgreen]{hyperref}

\definecolor{webgreen}{rgb}{0,.5,0}
\definecolor{webbrown}{rgb}{.6,0,0}

\usepackage{color}
\usepackage{fullpage}
\usepackage{float}

\usepackage{psfig}
\usepackage{graphics,amsmath,amssymb}
\usepackage{amsthm}
\usepackage{amsfonts}
\usepackage{latexsym}
\usepackage{epsf}

\usepackage{enumerate}

\begin{document}

%\begin{center}
%\epsfxsize=4in
%\leavevmode\epsffile{logo129.eps}
%\end{center}

\theoremstyle{plain}
\newtheorem{theorem}{Theorem}[section]
\newtheorem{corollary}[theorem]{Corollary}
\newtheorem{lemma}[theorem]{Lemma}
\newtheorem{proposition}[theorem]{Proposition}

\theoremstyle{definition}
\newtheorem{definition}[theorem]{Definition}
\newtheorem{example}[theorem]{Example}
\newtheorem{conjecture}[theorem]{Conjecture}

\theoremstyle{remark}
\newtheorem{remark}[theorem]{Remark}

\theoremstyle{plain}	\newtheorem{question}[theorem]{Question}

\begin{center}
\vskip 1cm{\LARGE\bf Towards a Statement of the $S$-adic\\ 
\vskip .17in
Conjecture through Examples}
\vskip 1cm
\large Fabien Durand\\
LAMFA, CNRS UMR 7352\\
Universit\'e de Picardie Jules Verne\\
UFR des Sciences\\
33, rue Saint-Leu\\
80039 Amiens Cedex 1, 
France\\
\href{mailto:fabien.durand@u-picardie.fr}{\tt fabien.durand@u-picardie.fr}\\
\ \\
Julien Leroy\\
Department of Mathematics\\
University of Li\`ege\\
Grande Traverse 12 (B37),\\
B-4000 Li\`ege, Belgium.\\
and\\
LAMFA, CNRS UMR 7352\\
Universit\'e de Picardie Jules Verne\\
UFR des Sciences\\
33, rue Saint-Leu\\
80039 Amiens Cedex 1, 
France\\
\href{mailto:j.leroy@ulg.ac.be}{\tt j.leroy@ulg.ac.be}\\
\ \\
Gw\'ena\"el Richomme\\
Universit\'e Paul-Val\'ery Montpellier 3\\
UFR IV, Dpt MIAp, Case J11,\\
Route de Mende,\\
34199 Montpellier Cedex 5, France\\
and\\
LIRMM (CNRS, Univ. Montpellier 2) - UMR 5506 - CC 477,\\
161 rue Ada, 34095,\\ Montpellier Cedex 5, France\\
\href{mailto:gwenael.richomme@lirmm.fr}{\tt gwenael.richomme@lirmm.fr}\\
\end{center}

\vskip .2 in
\begin{abstract}
The $S$-adic conjecture claims that there exists a condition $C$ such that a sequence has a sub-linear complexity if and only if it is an $S$-adic sequence satisfying Condition $C$ for some finite set $S$ of morphisms. 
We present an overview of the factor complexity of $S$-adic sequences and we give some examples that either illustrate some interesting properties or that are counter-examples to what could be believed to be ``a good Condition $C$''. 
\end{abstract}

\section{Introduction}

A usual tool in the study of a sequence (or an infinite word) $\mathbf{w}$ over a finite alphabet $A$ is the complexity function $p_\mathbf{w}$ (or simply $p$) that counts the number of factors of each length $n$ occurring in the sequence, \textit{i.e.}, $p_{\mathbf{w}}(n) = {\rm Card}(\{u \in A^* \mid |u|=n \ \text{and}\ \exists x \in A^*, \mathbf{y} \in A^{\mathbb{N}}: \mathbf{w} = xu\mathbf{y}\})$ (see Chapter 4 of~\cite{CANT} for a survey on complexity function). The set of factors of length $n$ of $\mathbf{w}$ is denoted by ${\rm Fac}_n(\mathbf{w})$ and ${\rm Fac}(\mathbf{w}) = \bigcup_{n \in \mathbb{N}} {\rm Fac}_n(\mathbf{w})$. The complexity function $p_{\mathbf{w}}$ is clearly bounded by $({\rm Card}(A))^n, n \in \mathbb{N}$, but not any function can be a complexity function. As an example, it is well known (see~\cite{Morse-Hedlund}) that either the sequence is ultimately periodic (and then $p_\mathbf{w}(n)$ is ultimately constant), or its complexity function grows at least like $n+1$. Non-periodic sequences with minimal complexity $p_\mathbf{w}(n) = n+1$ for all $n$ exist. They are called \textit{Sturmian sequences} and a large bibliography is devoted to them (see Chapter 2 of~\cite{Lothaire} and Chapter 6 of~\cite{Pytheas-Fogg} for surveys on these sequences). 

There is a huge literature about sequences with a low complexity. Indeed,  see for instance~\cite{Aberkane,Adamczewski,allouche_survey,Arnoux-Rauzy,Cassaigne_big_thm,Cassaigne_resume,chacon,ferenczi_survey,Glen-Justin,koskas,Rote}. By ``low complexity'' we usually mean that ``the complexity is bounded by a linear function''. Moreover, many well-known families of sequences can also be indefinitely desubstituted with a finite number of morphisms. Formally, an \textit{$S$-adic sequence} is defined as follows. Let $\mathbf{w}$ be a sequence over a finite alphabet $A$. If $S$ is a set of morphisms (possibly infinite), an \textit{$S$-adic representation} of $\mathbf{w}$ is given by a sequence $(\sigma_n: A_{n+1}^* \rightarrow A_n^*)_{n \in \mathbb{N}}$ of morphisms in $S$ and a sequence $(a_n)_{n \in \mathbb{N}}$ of letters, $a_i \in A_i$ for all $i$, such that\footnote{{The topology over $A^{\mathbb{N}}$ is the classical product topology of the discrete topology over $A$.}} $A_0 = A$, $\mathbf{w} = \lim_{n \rightarrow +\infty} \sigma_0 \sigma_1 \cdots \sigma_n (a_{n+1}^{\omega})$ {and $\lim_{n \rightarrow +\infty} |\sigma_0 \sigma_1 \cdots \sigma_n (a_{n+1})|=+\infty$, where $a_{n+1}^{\omega}$ is the sequence only composed of occurrences of $a_{n+1}$}. The sequence $(\sigma_n)_{n \in \mathbb{N}} \in S^{\mathbb{N}}$ is the \textit{directive word} of the representation. In the sequel, we will say that a sequence $\mathbf{w}$ is \textit{$S$-adic} if there exists a set $S$ of morphisms such that $\mathbf{w}$ admits an $S$-adic representation. Observe that in most cases, we deal with finite sets $S$ of morphisms.

An open problem is to determine the link between being an $S$-adic sequence and having a sub-linear complexity (see~\cite{Arnoux-Rauzy,Ferenczi,Leroy}). This problem is called the \textit{$S$-adic conjecture}.
\begin{conjecture}[$S$-adic conjecture]
There exists a condition $C$ such that a sequence has a sub-linear complexity if and only if it is an $S$-adic sequence satisfying Condition $C$ for some finite set $S$ of morphisms.
\end{conjecture} 
It is clear that we cannot avoid considering a particular condition since there exist some purely substitutive sequences with a quadratic complexity.

In this paper, we present an overview of the factor complexity of $S$-adic sequences and we give some examples that either illustrate some interesting properties or that are counter-examples to what could be believed to be ``a good Condition $C$''. 

In all what follows, we consider that alphabets are finite subsets of $\mathbb{N}$ and if $\sigma : A^* \to B^*$ is a morphism with $A = \{0,1,\dots,k\}$, we write $\sigma = [\sigma(0), \dots,\sigma(k)]$. The following example is classical when considering $S$-adic sequences.

\begin{example}
\label{ex: sturmian}
Let us define the four morphisms $R_0$, $R_1$, $L_0$ and $L_1$ over $\{0,1\}$ by $R_0 = [0,10]$, $R_1 = [01,1]$, $L_0 = [0,01]$ and $L_1 = [10,1]$. Since the work of Morse and Hedlund~\cite{Morse-Hedlund}, it is well known that for any Sturmian sequence $\mathbf{w}$, there is a sequence $(k_n)_{n \in \mathbb{N}}$ of integers such that
\begin{equation}
\label{eq: sturmian}
	\mathbf{w} = \lim_{n \to +\infty} L_0^{k_0} R_0^{k_1} L_1^{k_2} R_1^{k_3} L_0^{k_4} R_0^{k_5} \cdots L_1^{k_{4n+2}} R_1^{k_{4n+3}} (0^{\omega}).
\end{equation}
\end{example}

It is important to notice that, when we talk about an $S$-adic sequence, the corresponding directive word $(\sigma_n)_{n \in \mathbb{N}} \in S^{\mathbb{N}}$ is always implicit (even when it is not unique). Indeed, for a given set $S$ of morphisms, we will see that two distinct $S$-adic sequences can have different properties depending on their respective directive words.

\section{Comparison between morphic and $S$-adic sequences}

The aim of this section is to compare morphic sequences with $S$-adic sequences. In particular, we show that the factor complexity of morphic sequences is rather restricted and can be dependent on some combinatorial criteria although it is not the case at all for $S$-adic sequences.

\subsection{Morphic and purely morphic sequences}

\textit{Purely morphic sequences} correspond to $S$-adic sequences with ${\rm Card}(S)=1$. If $S=\{\sigma\}$, we then have $\sigma_0 \sigma_1 \sigma_2 \cdots = \sigma \sigma \sigma \cdots = \sigma^{\omega}$. In that case, the complexity functions that can occur have been completely determined by Pansiot in~\cite{Pansiot}. Indeed, he proved that for purely morphic sequences $\mathbf{w} = \sigma^{\omega}(a) = \lim_{n \to +\infty} \sigma^n(a^{\omega})$ with $\sigma$ \textit{non-erasing} (\textit{i.e.}, $\sigma(b)$ is not the empty word for all letters $b$), the complexity function $p_{\mathbf{w}}(n)$ can have only five asymptotic behaviors that are\footnote{$f(n) \in \Theta(g(n))$ if $\exists C_1, C_2 >0, n_0 \ \forall n \geq n_0 \ |C_1 g(n)| \leq |f(n)| \leq |C_2 g(n)|$.} $\Theta(1)$, $\Theta(n)$, $\Theta(n \log n)$, $\Theta(n \log \log n)$ and $\Theta(n^2)$. Moreover, when the sequence $\mathbf{w}$ is aperiodic, Morse and Hedlund proved that its complexity function cannot be $\Theta(1)$ (see~\cite{Morse-Hedlund}). Then, Pansiot proved that the class of complexity of the sequence only depends on the growth rate of the length of the images.

\begin{definition}
\label{def: quasi-uniforme, poly div, expo div}
Recall that a morphism $\sigma: A^* \to A^*$ is said to be \textit{everywhere growing} if it does not admit \textit{bounded letter}, i.e., letter $b$ such that $\lim_{n \to +\infty}|\sigma^n(b)|<+\infty$. We let $A_{\mathfrak{B},\sigma}$ (or $A_{\mathfrak{B}}$ when no confusion is possible) denote the set of bounded letters of $\sigma$. By opposition, a non-bounded letter is called a \textit{growing} letter. Since for all letters $a$, we have $|\sigma^n(a)| \in \Theta(n^{\alpha_a} \beta_a^n)$ for some $\alpha_a$ in $\mathbb{N}$ and $\beta_a \geq 1$ (see~\cite{L_system}), any everywhere growing morphism satisfies exactly one of the following three definitions:
\begin{enumerate} 
	\item	a morphism $\sigma: A^* \to A^*$ is \textit{quasi-uniform} if there exists $\beta \geq 1$ such that for all letters $a \in A$, $|\sigma^n(a)| \in \Theta(\beta^n)$;
	
	\item	a morphism $\sigma: A^* \to A^*$ is \textit{polynomially diverging} if there exists $\beta >1$ and a function $\alpha: A \to \mathbb{N}$, $\alpha \neq 0$, such that for all letters $a \in A$, $|\sigma^n(a)| \in \Theta(n^{\alpha(a)} \beta^n)$;
	
	\item 	a morphism $\sigma: A^* \to A^*$ is \textit{exponentially diverging} if there exist $a_1,a_2 \in A$, $\alpha_1,\alpha_2 \in \mathbb{N}$ and $\beta_1,\beta_2 >1$ with $\beta_1 \neq \beta_2$ such that for each $i \in \{1,2\}$, $|\sigma^n(a_i)| \in \Theta(n^{\alpha_i} \beta_i^n)$.
\end{enumerate}
\end{definition}

\begin{theorem}[Pansiot~\cite{Pansiot}]
\label{thm: pansiot}
Let $\mathbf{w} = \sigma^{\omega}(a)$ be a purely morphic sequence with $\sigma$ non-erasing.
\begin{enumerate}
	\item 	If $\sigma$ is everywhere growing and
		\begin{enumerate}[i.]
			\item quasi-uniform, then\footnote{$f(n) \in O(g(n))$ if $\exists C >0, n_0 \ \forall n \geq n_0 \ |f(n)| \leq |C g(n)|$.} $p_{\mathbf{w}}(n) \in O(n)$;
			\item polynomially diverging, then $p_{\mathbf{w}}(n) \in \Theta(n \log \log n)$;
			\item exponentially diverging, then $p_{\mathbf{w}}(n) \in \Theta(n \log n)$.
		\end{enumerate}
	\item If $\sigma$ is not everywhere growing and if there are infinitely many factors of $\mathbf{w}$ in $A_{\mathfrak{B}}^*$, then $p_{\mathbf{w}}(n)=\Theta(n^2)$.
	\item If $\sigma$ is not everywhere growing and if there are only finitely many factors of $\mathbf{w}$ in $A_{\mathfrak{B}}^*$, then there exists a purely morphic sequence $\tau^{\omega}(b)$ with $\tau: B^* \to B^*$ everywhere growing and a non-erasing morphism $\lambda: B^* \to A^*$ such that $\mathbf{w} = \lambda (\tau^{\omega}(b))$. In this case, we have $p_{\mathbf{w}}(n) \in \Theta(p_{\tau^{\omega}(b)}(n))$.  
\end{enumerate}
\end{theorem}

One could regret that Theorem~\ref{thm: pansiot} only holds for non-erasing morphisms. However, the following result states that when the morphism is erasing, one can see the obtained purely morphic sequence as a \textit{morphic sequence} (\textit{i.e.}, an image under a morphism of a purely morphic sequence) with non-erasing morphisms. The result is due to Cobham~\cite{Cobham_hartmanis} and has been recovered later by Pansiot~\cite{Pansiot_hierarchical}. It can also be found in Cassaigne and Nicolas's survey~\cite{Cassaigne-Nicolas}. 

\begin{theorem}[Cobham~\cite{Cobham_hartmanis} and Pansiot~\cite{Pansiot_hierarchical}]
\label{thm: cobham-erasing}
If $\mathbf{w}$ is a morphic sequence, it is the image under a letter-to-letter morphism of a purely morphic word $\sigma^{\omega}(a)$ with $\sigma$ a non-erasing morphism.
\end{theorem}

Theorems~\ref{thm: pansiot} and~\ref{thm: cobham-erasing} show that to compute the complexity function of a purely morphic sequence, it is sometimes necessary to see it as a morphic sequence. It is therefore natural to be interested in the complexity function of such sequences. By definition, it is obvious that any purely morphic sequence is morphic. The converse is known to be false since at least 1980, when Berstel proved that the Arshon word is morphic but not purely morphic \cite{Berstel1980} (other examples can be found in the litterature, as for instance, the result by Seebold showing that the unique, up to letter permutation, binary overlap-free word which is a fixed point of a morphism is the Thue-Morse word \cite{Seebold1985}).

Moreover, not only the class of morphic sequences strictly contains the class of purely morphic sequences, but also the asymptotic behaviors of the complexity functions are different. Indeed, Example~\ref{ex: complexite n^{3/2}} shows that the classes of complexity given by Pansiot are not sufficient anymore.

\begin{example}[Deviatov~\cite{Deviatov}]
\label{ex: complexite n^{3/2}}
Let $\mathbf{w}$ be the morphic sequence $\tau(\sigma^{\omega}(0))$ where $\sigma$ and $\tau$ are defined by
\begin{eqnarray*}
	\sigma :	\begin{cases}
					0 \mapsto 01 \\
					1 \mapsto 12 \\
					2 \mapsto 23 \\
					3 \mapsto 3
				\end{cases}
	&
	\text{and}
	&
	\sigma :	\begin{cases}
					0 \mapsto 0 \\
					1 \mapsto 1 \\
					2 \mapsto 2 \\
					3 \mapsto 2
				\end{cases}
\end{eqnarray*}
We have $p_{\mathbf{w}} \in \Theta(n \sqrt{n})$.
\end{example}

Other examples can be found in~\cite{Pansiot_subword}. Indeed, for all $k \geq 1$, Pansiot explicitly built a morphic sequence $\mathbf{w}$ whose complexity function satisfies $p_{\mathbf{w}}(n) \in \Theta(n \sqrt[k\,]{n})$. Consequently, the number of different asymptotic behaviors for the complexity function of morphic sequences is at least countably infinite. However, the behaviors $\Theta(n \sqrt[k\,]{n})$ seem to be the only new behaviors with respect to purely morphic sequences. Indeed, in~\cite{Deviatov} Deviatov proved the next result and conjectured an equivalent result of Pansiot's Theorem (Theorem~\ref{thm: pansiot}) for morphic sequences.

\begin{theorem}[Deviatov~\cite{Deviatov}]
\label{thm: Deviatov}
Let $\mathbf{w}$ be a morphic sequence. Then, either $p_{\mathbf{w}}(n) \in \Theta(n^{1+\frac{1}{k}})$ for some $k \in \mathbb{N}^*$, or $p_{\mathbf{w}}(n) \in O(n \log n)$. 
\end{theorem}

\begin{conjecture}[Deviatov~\cite{Deviatov}]
The complexity function of any morphic sequence only adopts one of the following asymptotic behaviors: $\Theta(1)$, $\Theta(n)$, $\Theta(n \log \log n)$, $\Theta(n \log n)$, $\Theta(n^{1 + \frac{1}{k}})$ for some $k \in \mathbb{N}$.
\end{conjecture}

In particular, Theorem~\ref{thm: Deviatov} implies that the highest complexity that one can get is the same for morphic sequences and for purely morphic sequences. This can be explained by the following result. 

\begin{proposition}[Cassaigne and Nicolas~\cite{Cassaigne-Nicolas}]
\label{prop: cassaigne nicolas action morphisme sur une suite}
Let $\mathbf{w}$ be a one-sided sequence over $A$ and $\sigma:A^* \to B^*$ be a non-erasing morphism. If $M = \max_{a \in A} |\sigma(a)|$, for all $n$ we have $p_{\sigma(\mathbf{w})}(n) \leq M p_{\mathbf{w}}(n)$. Moreover, if $\mathbf{w}$ is purely morphic and $\sigma$ is injective, then $p_{\sigma(\mathbf{w})}(n) \in \Theta(p_{\mathbf{w}}(n))$.
\end{proposition}

{The previous discussion shows that the factor complexity of morphic sequences is rather constrained. To conclude this section, we give some examples of how some additional combinatorial criteria can even more restrict it.}

A first well-known fact is that if a purely morphic sequence is \textit{$k$-power-free}, \textit{i.e.}, it does not contains any factor of the form $u^k$, then its factor complexity grows at least linearly and at most like $n \log n$ (see~\cite{D0L_m-free}). {We will consider a similar criterion on $S$-adic sequences in Section}~\ref{subsection: distinct powers}.

{Another such criterion is the \textit{uniform recurrence} of the sequence, \textit{i.e.}, any factor occurs infinitely often and with bounded gaps in $\mathbf{w}$. For morphic sequences, this implies that the complexity is sub-linear }(see~\cite{Nicolas-Pritykin}).{ Actually, the uniform recurrence of a morphic sequence is even equivalent to its \textit{linear recurrence}, \textit{i.e.}, any factor $u$ occurs infinitely often and with gaps bounded by $K|u|$} (see~\cite{Durand_UR, Durand_preprint}). {Furthermore, if $\mathbf{w}$ is a morphic and uniformly recurrent sequence over $A$, then $\mathbf{w}$ is a morphic sequence $\tau (\sigma^{\omega}(a))$ with $\sigma$ a primitive morphism (this result was already proved in}~\cite{Durand_UR} {in the particular case of morphic sequences $\psi(\varphi^{\omega}(a))$ with $\psi$ non-erasing and $\varphi$ everywhere growing). The following result also provides an algorithm to check whether a purely morphic sequence is uniformly recurrent.}

\begin{theorem}[Damanik and Lenz~\cite{Damanik-Lenz}]
\label{thm: damanik lenz LR}
A purely morphic sequence $\mathbf{w} = \sigma^{\omega}(a)$ is uniformly recurrent if and only if there is a growing letter $b \in A$ that occurs with bounded gaps in $\mathbf{w}$ and such that for all letters $c \in A$ there is a power $\sigma^k$ such that $c$ occurs in $\sigma^k(b)$.
\end{theorem}

\subsection{$S$-adic sequences}
\label{subsection S-adic}

The previous section shows that the factor complexity of morphic sequences are rather restricted (especially for purely morphic sequences). In particular, the class of uniformly recurrent morphic sequences is strictly contained in the class of morphic sequences with sub-linear complexity. In this section, we show that, for $S$-adic sequences, things are strongly different. A first important result is the following.

\begin{proposition}[Cassaigne~\cite{Cassaigne_S-adic}]
\label{prop:Cassaigne adique}
Let $A$ be an alphabet and $l \notin A$. There exists a finite set $S$ of morphisms over $A' = A \cup \{l\}$ such that any sequence over $A$ is $S$-adic.
\end{proposition}

In particular, this implies that one can get any high complexity with $S$-adic sequences (although for morphic sequences, the complexity is at most quadratic). Moreover, the following proposition implies that the set of possible asymptotic behaviors for the complexity function of $S$-adic sequences is uncountable (although it is finite for purely morphic sequences and conjectured to be countable for morphic ones). {See also}~\cite{Mauduit-Moreira} {for another approach to build sequences whose complexity is closed to a given function.}

\begin{proposition}[Cassaigne~\cite{Cassaigne_intermediate}]
\label{prop: cassaigne complexite arbitraire}
Let $f : \mathbb{R}^+ \to \mathbb{R}^+$ be a function such that
\begin{enumerate}[i.]
	\item $\lim_{t \to +\infty} \frac{f(t)}{\log t} = +\infty$;
	\item $f$ is differentiable, except possibly at $0$;
	\item $\lim_{t \to +\infty} f'(t) t^{\beta} = 0$ for some $\beta > 0$;
	\item $f'$ is decreasing.
\end{enumerate}
Then there exists a uniformly recurrent sequence $\mathbf{w}$ over $\{0,1\}$ such that\footnote{$f(n) \sim g(n)$ if $\forall \varepsilon>0 \ \exists n_0 \ \forall n>n_0 \ |f(n)/g(n)-1|<\varepsilon$.} $\log(p_{\mathbf{w}}(n)) \sim f(n)$.
\end{proposition}

In particular, the function $f(n)$ in the previous proposition can be taken equal to $n^{\alpha}$ for any $\alpha$ with $0 < \alpha < 1$.

Another big difference is that the class of uniformly recurrent $S$-adic sequences with sub-linear complexity is a very short part of the class of uniformly recurrent $S$-adic sequences. Indeed, recall that the \textit{topological entropy} of a sequence over an alphabet $A$ is the real number $h$ with $0 \leq h \leq \log({\rm Card}(A))$ defined by 
\[
	h = \lim_{n \to \infty} \frac{\log (p(n))}{n}.
\]
A uniformly recurrent sequence $\mathbf{w}$ over an alphabet $A$ with at least two letters $a,b \in A$ cannot have maximal complexity ($p_\mathbf{w}(n) = {\rm Card}(A)^n$): since all powers $a^n$ occur in $\mathbf{w}$, there are unbounded gaps between two successive occurrences of $b$ in $\mathbf{w}$. However, together with Proposition~\ref{prop:Cassaigne adique}, the following result shows that, except the maximal one $p_\mathbf{w}(n) = {\rm Card}(A)^n$, any high complexity can be reached by uniformly recurrent $S$-adic sequences.

\begin{theorem}[Grillenberger~\cite{Grillenberger}]
\label{thm: grillenberger}
Let $A$ be an alphabet with $d = {\rm Card}(A) \geq 2$ and $h \in \left[ 0, \log(d)\right[$. There exists a uniformly recurrent one-sided sequence $\mathbf{w}$ over $A$ with topological entropy $h$. 
\end{theorem}

The following result provides an $S$-adic characterization of uniformly recurrent sequences. 
A sequence is said to be \textit{primitive $S$-adic (with constant $s_0$)} if $s_0 \in \mathbb{N}$ is such that for all $r \in \mathbb{N}$, all letters in $A_r$ occur in all images $\sigma_r \cdots \sigma_{r+s_0}(a)$, $a \in A_{r+s_0+1}$. 
It is said  to be \textit{weakly primitive $S$-adic}  if for all $r \in \mathbb{N}$, there exists $s > 0$ such that all letters in $A_r$ occur in all images $\sigma_r \cdots \sigma_{r+s}(a)$, $a \in A_{r+s+1}$. Any primitive $S$-adic sequence is weakly primitive and Durand~\cite{Durand_corrigentum} proved that any weakly primitive $S$-adic sequence is uniformly recurrent. {Following a construction based on return words (see Section}~\ref{subsection return words} {for definition), one can prove that the converse it true.} 
More precisely, we got

\begin{theorem}[Durand~\cite{Durand_corrigentum}, Leroy~\cite{Leroy_these}]
\label{thm: characterization S-adicity minimal}
A sequence $\mathbf{w}$ is uniformly recurrent if and only if it is weakly primitive $S$-adic.

Moreover, the directive word can be chosen to be proper and if $\mathbf{w}$ does not have a sub-linear complexity, then 
$S$ is infinite.
\end{theorem}

The proof is similar to the proof of the following result which gives an $S$-adic characterization of linearly recurrent sequences. Let recall that a sequence is said to be \textit{proper $S$-adic} if all morphisms in its directive word are \textit{proper}, \textit{i.e.}, for all $n$ there are letters $a,b \in A_n$ such that $\sigma_n(A_{n+1}) \subset a A_n^* b$.

\begin{theorem}[Durand~\cite{Durand_corrigentum}]
\label{thm: durand LR S-adic}
A sequence $\mathbf{w}$ is linearly recurrent if and only if it is primitive and proper $S$-adic with ${\rm Card}(S)< +\infty$.
\end{theorem}

The next example shows that the linear recurrence is an even stronger condition that the primitivity for $S$-adic sequences  (contrary to what holds in the morphic case).

\begin{example}[Durand~\cite{Durand_corrigentum}]
\label{ex: durand primitif S-adic pas LR}
Let $S = \{\sigma,\tau\}$ where $\sigma$ and $\tau$ are defined by
\begin{eqnarray*}
	\sigma:	\begin{cases}
				0 \mapsto 021	\\
				1 \mapsto 101	\\
				2 \mapsto 212
			\end{cases}
	&
	\text{and}
	&
	\tau:	\begin{cases}
				0 \mapsto 012	\\
				1 \mapsto 021	\\
				2 \mapsto 002
			\end{cases}	
\end{eqnarray*}
The sequence
\[
	\mathbf{w} = \lim_{n \to +\infty} \sigma \tau \sigma^2 \tau \cdots \sigma^n \tau (0^{\omega})
\]
is primitive $S$-adic but not linearly recurrent.
\end{example}

\section{$S$-adicity and sub-linear complexity}

The aim of this section is to explore some ideas that one could have about the $S$-adic conjecture. First, we present some sufficient conditions for $S$-adic sequences to have a sub-linear complexity and we show that we cannot make them weaker. Then, we give some counter-examples to some conditions that one could naturally believe to be sufficient to have a sub-linear complexity.

\subsection{The growth rate of the length have still some importance}
\label{subsection croissance}

Durand~\cite{Durand_LR,Durand_corrigentum}, gave some sufficient conditions for an $S$-adic sequence to have a sub-linear complexity. The main condition is the one given by the following result and is a generalization of what exists for purely morphic sequences (see Theorem~\ref{thm: pansiot}). The other conditions are simply consequences of it.

\begin{proposition}[Durand~\cite{Durand_corrigentum}]
\label{prop: durand longueur comparable S-adic}
Let $\mathbf{w}$ be an $S$-adic sequence with ${\rm Card}(S)<+\infty$ and whose directive word is $(\sigma_n)_{n \in \mathbb{N}}$ with $\sigma_n: A_{n+1}^* \to A_n^*$ and $A_0$ the alphabet of $\mathbf{w}$. If there is a constant $D$ such that for all $n$,
\begin{equation}
\label{cond croissance S-adic}
	\max_{a,b \in A_{n+1}} \frac{|\sigma_0 \cdots \sigma_n(a)|}{|\sigma_0 \cdots \sigma_n(b)|} \leq D,
\end{equation}
then $p_{\mathbf{w}}(n) \leq D \max_{\sigma_n \in S, a \in A_{n+1}} |\sigma_n(a)| ({\rm Card}(A))^2 n$ with $A = \cup_{n \in \mathbb{N}}A_n$.
\end{proposition}

\begin{corollary}[Durand~\cite{Durand_corrigentum}]
\label{cor: durand uniform S-adic}
If $\mathbf{w}$ is $S$-adic with ${\rm Card}(S)< \infty$ and all morphisms in $S$ are uniform, then we have $p_{\mathbf{w}}(n) \leq l ({\rm Card}(A))^2 n$ with $A = \cup_{n \in \mathbb{N}}A_n$ and $l = \max\limits_{\sigma \in S, a \in A(\sigma)} |\sigma(a)|$.
\end{corollary}
 
\begin{proposition}[Durand~\cite{Durand_LR}]
\label{prop: durand primitive S-adic}
If $\mathbf{w}$ is a primitive $S$-adic sequence with ${\rm Card}(S)<+\infty$ and constant $s_0$ directed by $(\sigma_n)_{n \in \mathbb{N}}$ with $\sigma_n: A_{n+1}^* \to A_n^*$ and $A_0$ the alphabet of $\mathbf{w}$, then there exists a constant $D$ such that for all non-negative integers $r$, 
\[
	\max_{a,b \in A_{r+s_0+1}}\frac{|\sigma_r \cdots \sigma_{r+s_0}(a)|}{|\sigma_r \cdots \sigma_{r+s_0}(b)|} \leq D.
\]
\end{proposition}
  
\begin{corollary}
\label{cor: durand 1 morphisme fortement primitif}
Let $S$ be a set of non-erasing morphisms and $\tau \in S$ be strongly primitive (\textit{i.e.}, for all letters $a$ of $\tau(A)$, $a$ occurs in all images $\tau(b)$ for $b \in A$). Any $S$-adic sequence for which $\tau$ occurs infinitely often with bounded gaps in the directive word is uniformly recurrent and has a sub-linear complexity.
\end{corollary}

The everywhere growing property can be naturally transposed to $S$-adic sequences. But, contrary to what holds in the purely morphic case, under that assumption, the condition given by Equation~\eqref{cond croissance S-adic} is not equivalent to sub-linear complexity. Indeed, even some Sturmian sequences do not satisfy it (those with unbounded coefficient $(k_n)_{n \in \mathbb{N}}$ in Example~\ref{ex: sturmian}). One could therefore try to make that condition a little bit weaker.

We can observe in Example~\ref{ex: sturmian} that Condition~\eqref{cond croissance S-adic} is still infinitely often satisfied. Indeed, let us consider the  directive word $(\tau_n)_{n \in \mathbb{N}}$ defined by $\tau_0 = L_0^{k_0} R_0^{k_1} L_1$, $\tau_{2n} = L_0^{k_{4n}-1} R_0^{k_{4n+1}} L_1$ for $n \geq 1$ and $\tau_{2n+1} = L_1^{k_{4n+2}-1} R_1^{k_{4n+3}} L_0$ for $n \geq 0$. For all $n$, there exist some integers $i$ and $j$ such that either $\tau_n = [0^i10^{j+1},0^i10^j]$ or $\tau_n = [1^i01^j,1^i01^{j+1}]$. With these morphisms, we have
\begin{equation}
\label{eq: sturmian contracte}
	\tau_0 \tau_1 \cdots \tau_n \cdots = L_0^{k_0} R_0^{k_1} L_1^{k_2} R_1^{k_3} \cdots L_1^{k_{4n+2}} R_1^{k_{4n+3}} \cdots
\end{equation}
and there is a constant $K$ such that for all $n$, $\max_{a,b \in A_{n+1}} \frac{|\tau_0 \cdots \tau_n(a)|}{|\tau_0 \cdots \tau_n(b)|} \leq K$. The sequence $(\tau_n)_{n \in \mathbb{N}}$ is called a \textit{contraction} of the directive word $(L_0^{k_0} R_0^{k_1} \cdots)$. Observe that the set $S = \{\tau_n \mid n \in \mathbb{N}\}$ of morphisms might be infinite (when $(k_n)_{n \in \mathbb{N}}$ is unbounded). Consequently, it may be interesting to work either with infinite sets of morphisms or with contractions.

But, Example~\ref{ex: morse+n^2} below shows that Proposition~\ref{prop: durand longueur comparable S-adic} is not true anymore when ${\rm Card}(S)=\infty$. Indeed, if we consider the  contraction $(\sigma_n)_{n \in \mathbb{N}}$ of the directive word of Proposition~\ref{prop: morse+n^2} defined for all $n \geq 0$ by 
\begin{equation}
\label{eq: morse+n^2 contr.acte}
	\sigma_n = \gamma^{k_n} \mu,
\end{equation}
we have $|\sigma_0 \cdots \sigma_n(0)| = |\sigma_0 \cdots \sigma_n(1)|$ for all $n$ although the complexity is not sub-linear as soon as the sequence $(k_n)_{n \in \mathbb{N}}$ is unbounded.

\begin{example} 
\label{ex: morse+n^2}
Let $\mu$ be the \textit{Thue-Morse morphism} $[01,10]$ and let $\gamma$ be the morphism $[001,1]$. From Theorem~\ref{thm: pansiot} we know that the sequence 
\[
	\gamma^{\omega}(0) = 001 001^2 001 001^3 001 001^2 001 001^4 \cdots  
\]
has a quadratic complexity.
 
\begin{proposition}
\label{prop: morse+n^2}
Let $(k_n)_{n \in \mathbb{N}}$ be a sequence of non-negative integers. The sequence 
\[
	\mathbf{w}_{\gamma,\mu} = \lim_{n \to +\infty} \gamma^{k_0} \mu \gamma^{k_1} \mu \gamma^{k_2} \mu \cdots \gamma^{k_n} \mu (0^{\omega})
\]
is uniformly recurrent. Moreover, $\mathbf{w}_{\gamma,\mu}$ has an at most linear complexity if and only if the sequence $(k_n)_{n \in \mathbb{N}}$ is bounded. Finally, for all $n$ we have 
\[
	|\gamma^{k_0} \mu \gamma^{k_1} \mu \gamma^{k_2} \mu \cdots \gamma^{k_n} \mu(0)| = |\gamma^{k_0} \mu \gamma^{k_1} \mu \gamma^{k_2} \mu \cdots \gamma^{k_n} \mu(1)|,
\]
and denoting
\[
	\ell_n = |\gamma^{k_0} \mu \gamma^{k_1} \mu \gamma^{k_2} \mu \cdots \gamma^{k_n} \mu(0)|,
\]
we have
\[
	p_{\mathbf{w}_{\gamma,\mu}}(\ell_n) \leq 4 \ell_n-2.
\]
\end{proposition}

Before proving the result, recall that a \textit{right} (resp. \textit{left}) \textit{special factor} in a language $L \subset A^*$ is a factor such that there are at least two letters $a,b \in A$ for which $ua$ and $ub$ (resp. $au$ and $bu$) belong to $L$.  A \textit{bispecial factor is a factor which is both left and right special.} This definition can be extended to words and sequences $\mathbf{w}$ by replacing $L$ by ${\rm Fac}(\mathbf{w})$. Let us recall the following result.

\begin{theorem}[Cassaigne~\cite{Cassaigne_big_thm}]
\label{theorem: cassaigne diff de complexite}
A sequence has a sub-linear complexity if and only if there is a constant $K$ such that for all $n$, the number of right (resp. left) special factors of length $n$ is less than $K$.
\end{theorem}

\begin{proof}[Proof of Proposition~\ref{prop: morse+n^2}]
First, as $\mu$ occurs infinitely often in the directive word, 
$\mathbf{w}_{\gamma,\mu}$ is weakly primitive $\{\gamma, \mu\}$-adic and so, by Theorem~\ref{thm: characterization S-adicity minimal}, it is uniformly recurrent.

Now let us study the complexity depending on the sequence $(k_n)_{n \in \mathbb{N}}$. The case of a bounded sequence is a direct consequence of Corollary~\ref{cor: durand 1 morphisme fortement primitif}. Hence let us consider that the sequence $(k_n)_{n \in \mathbb{N}}$ is unbounded and let us show that the complexity is not at most linear. Due to Theorem~\ref{theorem: cassaigne diff de complexite}, we only have to prove that the number of right special factors of length $n$ of $\mathbf{w}_{\gamma,\mu}$ is unbounded.

As already mentioned in Example~\ref{ex: morse+n^2}, the fixed point $\gamma^{\omega}(0)$ has a quadratic complexity. Consequently the number of right special factors of $\gamma^{\omega}(0)$ of a length $n$ is unbounded. 

We let readers check that bispecial factors of $\gamma^\omega(0)$ are words $\varepsilon$, $0$, $1$ and words $\gamma(1v)$ with $v$ bispecial, and that other right special factors of $\gamma^\omega(0)$ are the suffixes of words $\gamma(v)$ for $v$ right special. Thus, by induction, one can state that all right special factors of length $n$ of 
 $\gamma^{\omega}(0)$ occur in $\gamma^{n+1}(0)$. 
 
Now let us show that if $u$ is a right special factor in $\gamma^{k_{n+1}}(0)$, then $\gamma^{k_0} \mu \gamma^{k_1} \mu \cdots \gamma^{k_n} \mu (u)$ is a right special factor of $\mathbf{w}$ of length $|u| 2^{q}$ with $q = \sum_{i = 0}^{n} (k_i +1)$. Indeed, as $\mu(0)$ and $\gamma(0)$ start with $0$ and $\mu(1)$ and $\gamma(1)$ start with $1$, the image of $u$ is still a right special factor. Moreover, $\mu(u)$ contains exactly $|u|$ occurrences of the letter $0$ and $n$ occurrences of the letter $1$, and both $\gamma$ and $\mu$ map a word with the same number of $0$ and $1$ to a word of double length with the same number of $0$ and $1$. Hence $|\gamma^{k_0} \mu \gamma^{k_1} \mu \cdots \gamma^{k_n} \mu (u)| = |u| 2^q$ with $q$ defined as previously. Now, if $u$ and $v$ are two distinct right special factors of length $m$ of $\gamma^{\omega}(0)$, then $\gamma^{k_0} \mu \gamma^{k_1} \mu \cdots \gamma^{k_n} \mu (u)$ and $\gamma^{k_0} \mu \gamma^{k_1} \mu \cdots \gamma^{k_n} \mu (v)$ are two distinct special factors of length $m 2^q$ of $\mathbf{w}$. As the number of right special factors of a given length of $\gamma^{\omega}(0)$ is unbounded, and as the sequence $(k_n)_{n \in \mathbb{N}}$ is unbounded, the number of right special factors of a given length of $\mathbf{w}$ is also unbounded which concludes the first part of the proof.

The last step is to show that, for all integers $\ell_n$, we have $p_{\mathbf{w}_{\gamma,\mu}}(\ell_n) \leq 4 \ell_n$. For all non-negative integers $n$, we already know that 
\[
	|\gamma^{k_0} \mu \gamma^{k_1} \mu \cdots \gamma^{k_n} \mu(0)| = |\gamma^{k_0} \mu \gamma^{k_1} \mu \cdots \gamma^{k_n} \mu(1)| = \ell_n = 2^q
\]
with $q$ as defined previously by $\sum_{i = 0}^n (k_i +1)$. Consequently, all factors $u$ of length $\ell_n$ are factors of $|\gamma^{k_0} \mu \gamma^{k_1} \mu \cdots \gamma^{k_n} \mu(v)|$ for some word $v$ of length 2. As there are only 4 possible binary words of length 2 and as there are less than $\ell_n+1$ distinct factors of length $\ell_n$ in a word of length $2\ell_n$, we obtain $p_{\mathbf{w}_{\gamma,\mu}}(\ell_n) \leq 4\ell_n +4$. However, among the $4\ell_n +4$ words, both words $\gamma^{k_0} \mu \gamma^{k_1} \mu \cdots \gamma^{k_n} \mu(0)$ and $\gamma^{k_0} \mu \gamma^{k_1} \mu \cdots \gamma^{k_n} \mu(1)$ have been counted 4 times, hence $p_{\mathbf{w}_{\gamma,\mu}}(\ell_n) \leq 4\ell_n-2$.
\end{proof}
\end{example}

\subsection{The condition $C$ of the conjecture could not only concern the set $S$}
\label{subsection: condition C}

Since it seems hard to make the condition of Equation~\eqref{cond croissance S-adic} weaker, another idea is to determine new sufficient conditions that are independent from it. A first attempt in this direction was proposed by Boshernitzan, asking whether \textit{if $S$ contains only morphisms for which the fixed points (if exists) of their powers have sub-linear complexity, then any $S$-adic sequence has a sub-linear complexity}.

But, Boshernitzan eventually provided the following counter-example to that conjecture. Since we did not find any detailed proof of it, we provided it.

\begin{example}
\label{ex: n^2}
Let $\gamma$ and $E$ be the morphisms over $\{0,1\}$ respectively defined by $[001,1]$ and $[1,0]$. Observe that both morphisms $\gamma E$ and $E \gamma$ are \textit{primitive}. Consequently, they admit a power whose fixed points have sub-linear complexity. We consider the  sequence 
\[
	\mathbf{w}_{\gamma,E} = \lim_{n \to + \infty} \gamma E \gamma^2 E \gamma^3 E \cdots \gamma^{n-1} E \gamma^n(0^{\omega}).
\]

\begin{proposition}[Boshernitzan]
\label{prop: n^2}
The sequence $\mathbf{w}_{\gamma,E}$ is $S$-adic for $S = \{ \gamma E, E \gamma\}$, is uniformly recurrent and does not have a sub-linear complexity.
\end{proposition}

\begin{proof}
First, by definition, $\mathbf{w}_{\gamma,E}$ is indeed $S$-adic for $S= \{ \gamma E, E \gamma\}$. 

Next, the composition $\gamma \circ E \circ \gamma$ is strongly primitive and occurs infinitely often in the directive word of $\mathbf{w}_{\gamma,E}$. It is therefore a consequence of Theorem~\ref{thm: characterization S-adicity minimal} that $\mathbf{w}_{\gamma,E}$ is uniformly recurrent. 

To prove that $\mathbf{w}_{\gamma,E}$ does not have a sub-linear complexity, by Theorem~\ref{theorem: cassaigne diff de complexite}, it is sufficient to prove that the number of its right special factors of length $n$ is unbounded. For this purpose, let us introduce notation. For all $k \in \mathbb{N}^*$, let us define the morphism $\Gamma_k = \gamma E \gamma^2 E \cdots \gamma^{k-1} E \gamma^k E$ and, for all $k \in \mathbb{N}$, the sequence
\[
	\mathbf{w}^{(k)} = \lim_{n \to + \infty} \gamma^{k+1} E \gamma^{k+2} E \cdots \gamma^{k+n-1} E \gamma^{k+n}(0^{\omega}).
\]

We then have $\mathbf{w}_{\gamma,E} = \mathbf{w}^{(0)} = \Gamma_k (\mathbf{w}^{(k)})$, for all $k \geq 1$. For all $i \geq 1$ we also consider the word $u_i = \gamma^i(10) = 1 \gamma^i(0)$. As $100$ and $101$ are factors of $E(\mathbf{w}^{(k+1)})$, and as $100$ and $001$ are factors of $\gamma^{k-j}\gamma(E(\mathbf{w}^{(k+1)}))$ for all $j$ with $1 \leq j \leq k$ (observe also that $\gamma(100)= \gamma(10)001$ and $\gamma(001) = 00\gamma(10)1$), we can deduce  that $u_i$ is a right special factor of $\mathbf{w}^{(k)}$ for all $i$ with $1 \leq i \leq k+1$.
As words $\Gamma_k(0)$ and $\Gamma_k(1)$ start with different letters, for all integers $i$ such that $1 \leq i \leq k+1$, the word $\Gamma_k(u_i)$ is a right special factor of $\mathbf{w}_{\gamma,E}$, and so are all of its suffixes.

For $1 \leq i \leq k+1$, word $u_i$ end with $1001^i$, so that the longest common suffix between $u_i$ and 
$u_{i+1}$ is $1^i$. It follows that the longest common suffix between $\Gamma_k(u_i)$ and $\Gamma_k(u_{i+1})$ is the word $p_k\Gamma_k(1^i)$ where $p_0 = \varepsilon$ and for $k \geq 1$,
\[
	p_k = 1 \Gamma_1(1^2)\Gamma_2(1^3)\cdots \Gamma_{k-1}(1^k).
\]
This is indeed a consequence of the fact (we let readers verify it) that for any word $u$ containing at least one occurrence of $0$ and $1$, the word $\Gamma^k(u)$ ends with $ap_k$ where $a$ is the last letter of $u$ when $k$ is even, and where $a$ is the opposite letter to the last letter of $u$ when $k$ is odd.

 From what precedes $\mathbf{w}_{\gamma,E}$ has at least $n(k)$ right special words of length $|p_k\Gamma_k(1^k)|+1$, where $n(k)$ denotes the number of integers $i$ between $1$ and $k+1$ such that $|\Gamma(u_i)| > |p_k\Gamma_k(1^k)|$. Next fact allows to estimate $n(k)$.

\textit{Fact}: With $f(k) = \frac{k^2+k+2}{2}$, for all $k \geq 2$,
\begin{enumerate}
\item for all $1 \leq i \leq k+1$, $|\Gamma_k(u_i)| = 2^i 2^{f(k)}$,
\item $|p_k\Gamma_k(1^k)| \leq k^2 2^{f(k)}$.
\end{enumerate}

\begin{proof}[Proof of the fact]

1. Let $i$ be between $1$ and $k+1$. By induction, one can verify that $|\gamma^i(0)|_0 = 2^i$ and  $|\gamma^i(0)|_1 = 2^i-1$. As $u_i = 1 \gamma^i(0)$, we have $|u_i|_0 = |u_i|_1 = 2^i$. Observe that for any word $v$ such that $|v|_0 = |v|_1$, we have $|\gamma(v)| = 2 |v|$ and $|\gamma(v)|_0 = |\gamma(v)|_1 = |v|$ and $|E(v)|_0 = |E(v)|_1 = |v|_0$.  

Thus, as  $\Gamma_k = \gamma E \gamma^2 E \cdots \gamma^{k-1} E \gamma^k E$, $|\Gamma_k(u_i)| = 2^{i+1} 2^{\sum_{j=1}^k j} = 2^{i+1} 2^{\frac{k(k+1)}{2}} = 2^i 2^{f(k)}$. 

2. Before estimating $|p_k\Gamma_k(1^k)|$, we need an estimate of $|\Gamma_k(1)|$.
\begin{eqnarray*}
	|\Gamma_k(1)| 	&=& \left| \gamma E \gamma^2 E \cdots \gamma^k E (1) \right|	\\
					&=&	\left| \gamma^k(0)\right|_0 \left| \gamma E \gamma^2 E \cdots \gamma^{k-1} E(0) \right| 	
						+ \left| \gamma^k(0)\right|_1 \left| \gamma E \gamma^2 E \cdots \gamma^{k-1} E(1) \right| \\
					&=& 2^k \left( \left| \Gamma_{k-1}(01) \right| - \left| \Gamma_{k-1}(1) \right| \right)
						+ \left( 2^k -1 \right) \left| \Gamma_{k-1}(1) \right|	\\
					&=&	2^{\frac{k^2+k+2}{2}} - \left| \Gamma_{k-1}(1) \right| < 2^{f(k)}.
\end{eqnarray*}
Observe that $|\Gamma_k(1^k)| = k |\Gamma_k(1)|$. One can verify that, for all $j$ with $1 \leq j \leq k$, $\Gamma_{j-1}(1^j) \leq \Gamma_k(1^k)$. Moreover as $k \geq 2$, $|1\Gamma_1(1^2)|\leq |\Gamma_k(1^k)|$.

Hence $|p_k\Gamma_k(1^k)| = |1\Gamma_1(1^2)|+\sum_{j=3}^{k}|\Gamma_{j-1}(1^j)|+|\Gamma_k(1^k)| < k|\Gamma_k(1^k)| < k^2f(k)$.
\end{proof}

To end the proof of proposition~\ref{prop: n^2}, it suffices to notice that for all $i$ such that $\log_2 k^2 < i \leq k+1$, 
$|\Gamma_k(u_i)| > |p_k\Gamma_k(1^k)|$, so that $n(k) \geq k+1- \left\lceil \log_2 k^2 \right\rceil$. In other words, the number of right special words of $\mathbf{w}_{\gamma,E}$ of length $|p_k\Gamma_k(1^k)|+1$ is unbounded, which shows by Theorem~\ref{theorem: cassaigne diff de complexite} that $\mathbf{w}_{\gamma,E}$ does not have a sub-linear complexity.
\end{proof}

\begin{remark}
The previous result is even stronger than just considering sets $S$ of morphisms with fixed points of sub-linear complexity. Indeed, the sequence also has \textit{bounded partial quotients}, \textit{i.e.}, all morphisms occur with bounded gaps in the directive word (over $\{\gamma E, E \gamma\}$). 
\end{remark}
\end{example}

An opposite question to the one previously answered is to ask whether $S$-adic sequences can have a sub-linear complexity when $S$ contains a morphism that admits a fixed point that does not have a sub-linear complexity. Example~\ref{ex: morse+n^2} positively answers that question and we can prove even more. Indeed, the following example provides a uniformly recurrent $S$-adic sequence with sub-linear complexity such that a ``bad morphism'' occurs with arbitrary high powers in its directive word.

\begin{example}
\label{ex: beta M}
Let us consider the morphisms
\begin{eqnarray*}
	\beta : \begin{cases}
				0 \mapsto 010	\\
				1 \mapsto 1112	\\
				2 \mapsto 2
			\end{cases}
	&
	\text{ and }
	&
	M :	\begin{cases}
			0 \mapsto 0	\\
			1 \mapsto 1 \\
			2 \mapsto 1
		\end{cases}
\end{eqnarray*}
and the sequence 
\[
	\mathbf{w}_{\beta,M} = \lim_{n \to +\infty} M \beta M \beta^2 M \beta^3 M \cdots \beta^{n-1} M \beta^n(0^{\omega}).
\]
Note that $\beta$ is not everywhere growing, its set of bounded letters is $A_{\mathfrak{B},\beta} = \{2\}$ and all words in $A_{\mathfrak{B},\beta}^*$ are factors of $\beta^\omega(0)$. Hence by Theorem~\mbox{\ref{thm: pansiot}(2)}, $p_{\beta^\omega(0)}$ is in $\theta(n^2)$.
\end{example}

\begin{proposition}
\label{prop: beta n^2}
The sequence $\mathbf{w}_{\beta,M}$ defined just above has a sub-linear complexity. More precisely, for all $n$ we have $p(n+1)-p(n) \in \{1,2\}$.
\end{proposition}

The proof of this proposition will use the link, proved by Cassaigne \mbox{\cite{Cassaigne_resume}} between bispecial factors and the first difference $s$ of the complexity function. Given an infinite word $\mathbf{w}$, we recall that the  \textit{bilateral order} of a word $u$ is the value
\[
 	m(u) = {\rm Card} ({\rm Fac}(\mathbf{w}_{\beta,M}) \cap A u A) - {\rm Card}({\rm Fac}(\mathbf{w}_{\beta,M}) \cap Au) - {\rm Card}({\rm Fac}(\mathbf{w}_{\beta,M}) \cap uA) +1.
\]
When $m(u)>0$, $u$ is a bispecial factor of $\mathbf{w}$ called a \textit{strong bispecial} factor, and
when $m(u)<0$, $u$ is also a bispecial factor of $\mathbf{w}$ called a \textit{weak bispecial} factor.

Let $s_\mathbf{w}$ (or simply $s$) be the function defined for $n \geq 0$ by $s(n) = p(n+1)-p(n)$. We have $s(0)=1$ and Cassaigne proved (see~\mbox{\cite{Cassaigne_resume}}) that
\begin{equation}
\label{eq:cassaigne}
	s(n+1) - s(n) = \sum_{u \in {\rm Fac}_n(\mathbf{w})} m(u)
\end{equation}

In next useful result, $\mathcal{B}_0$ is the identity morphism and, for $k > 0$, $\mathcal{B}_k = M \beta M \beta^2 \cdots M \beta^k$.

\begin{lemma}
\label{lem:w beta M}
A word $u$ is a strong bispecial factor of $\mathbf{w}_{\beta,M}$ if and only if $u = \mathcal{B}_k M \beta^i(1)$ for some $k \geq 0$ and some $i$ in $\{0, \ldots, k\}$.

A word $u$ is a weak bispecial factor of $\mathbf{w}_{\beta,M}$ if and only if $u = \mathcal{B}_k M \beta^i(101)$ for some $k \geq 0$ and some $i$ in $\{0, \ldots, k\}$.
\end{lemma}

\begin{proof}
For all integers $k \geq 0$,
let  $\mathbf{w}^{(k)}$ be the sequence $\mathbf{w}^{(k)}$ whose directive word is 
$(M,$ $\beta^{k+1},$ $M,$ $\beta^{k+2},\dots)$. Note that, for $k \geq 0$, $\mathbf{w}_{\beta,M} = \mathcal{B}_k(\mathbf{w}^{(k)})$.

All words $010$, $011$, $110$, $111$ occur in $\mathbf{w}^{(k+1)}$ and all words $010$, $011$, $120$, $111$ occur in $\beta^j(\mathbf{w}^{(k+1)})$ for $j \geq 1$. Thus $M\beta^i(1)$ for $i$ in $\{0, 1, \ldots, k+1\}$ are strong bispecial factors of $\mathbf{w}^{(k+1)}$ (observe $M\beta^i(120) = M\beta^i(210)$ contains the factor $1M\beta^i(1)0$). As $\mathcal{B}_k(0)$ starts and ends with $0$, and as $\mathcal{B}_k(1)$ starts and ends with $1$, $\mathcal{B}_k M\beta^i(1)$ for $i \in \{0, 1, \ldots, k\}$ are strong bispecial factors of $\mathbf{w}_{\beta,M}$.

Now let $i \in \{0, 1, \ldots, k\}$. Words $0 M\beta^i(101)1$ and $1M\beta^i(101)0$ are factors of $\mathbf{w}^{(k)}$ (respectively factors of $M\beta^{i+1}(01)1$ and $1M\beta^{i+1}(10)$). So $0\mathcal{B}_k(M\beta^i(101))1$ and $1 \mathcal{B}_k(M\beta^i(101)) 0$ are factors of $\mathbf{w}_{\beta,M}$.
One can verify that the existence of factor $0\mathcal{B}_k(M\beta^i(101))0$ in $\mathbf{w}_{\beta,M}$ would imply the existence of factors $M(0\beta^j(101)0)$ in $\beta^{(k+1-j)}(\mathbf{w}^{(k+1)})$ for $0 \leq j \leq i$, and of factor $01010$ in  $\beta^{(k+1-i)}(\mathbf{w}^{(k+1)})$ which is impossible. 
Similarly the existence of factor $1\mathcal{B}_k(M\beta^i(101))1$ in $\mathbf{w}_{\beta,M}$ would imply the existence of $11011$ in  $\beta^{(k+1-i)}(\mathbf{w}^{(k+1)})$ which is impossible.

We now prove that there are no other strong and weak bispecial factors.

Let $u$ be a bispecial factor of $\mathbf{w}_{\beta,M}$. If $u$ contains no occurrence of the letter $0$, $u = 1^n$ for some $n \geq 0$.
As $\mathbf{w}_{\beta,M}$ contains arbitrarily large powers of $1$ and infinitely many occurrences of $0$, words $11^n1$, $01^n1$, $11^n0$ are all factors of $\mathbf{w}_{\beta,M}$. When $01^n0$ is a factor of $\mathbf{w}_{\beta,M}$, one can verify that $1^n = \mathcal{B}_kM\beta^i(1)$ for some $k \geq 0$ and some integer $i$ in $\{0, \ldots, k\}$. Thus $1^n$ is a strong bispecial factor of $\mathbf{w}_{\beta,M}$ if and only if $1^n = \mathcal{B}_kM\beta^i(1)$ for some $k \geq 0$ and some integer $i$ in $\{0, \ldots, k\}$.

Note that any factor of $\mathbf{w}_{\beta,M}$ containing only one occurrence of $0$ is on the form $01^n$, $101^n$, $1^n0$ or $1^n01$ for some $n \geq 0$ and is not bispecial.

Thus assume $u$ contains at least two occurrences of $0$. From the shape of the morphisms $\mathcal{B}_k$, there is a unique integer $k$ and a unique sequence of words 
\[
	v, p_0, s_0, p_1, s_1, \dots, p_{k-1}, s_{k-1}
\]
over $\{0,1\}$ such that 
\[
	u = p_0 \mathcal{B}_1(p_1) \mathcal{B}_2(p_2) \cdots \mathcal{B}_{k-1}(p_{k-1}) \mathcal{B}_k(v) \mathcal{B}_{k-1}(s_{k-1}) \cdots \mathcal{B}_2(s_2) \mathcal{B}_1(s_1)s_0,
\]
where $v$ is a proper factor of $M\beta^{k+1}(010)$ and for all $i$, $0 \leq i \leq k-1$, the word $p_i$ (resp. $s_i$) is a suffix (resp. prefix) of $\mathcal{B}_{i+1}(0)$ or of $\mathcal{B}_{i+1}(1)$ but is neither equal to $\mathcal{B}_{i+1}(0)$ nor to $\mathcal{B}_{i+1}(1)$.

Since for all $i$, $\mathcal{B}_i(0)$ and $\mathcal{B}_i(1)$ do not have any common prefix or suffix, the factor $u$ is bispecial if and only if all words $p_0,s_0,p_1,s_1,\dots,p_{k-1},s_{k-1}$ are empty words and $v$ is bispecial in $\mathbf{w}^{(k)}$. Moreover, $v$ contains an occurrence of the letter $0$. We also have that $u$ is strongly (resp. weakly) bispecial if and only if so is $v$. Then, it is only a verification that the only bispecial factors in $\mathbf{w}^{(k)}$ that contain an occurrence of $0$ and that are factors of $M\beta^{k+1}(010)$ are the words $M \beta^i(101)$ for $i \in \{0,1,\dots,k\}$, which are weak bispecial factors.
\end{proof}

\begin{proof}[Proof of Proposition~\ref{prop: beta n^2}]
Let $(u_n)_{n \geq 0}$ be the sequence of strong and weak bispecial factors of $\mathbf{w}_{\beta,M}$.
Observe that for all $k \geq 0$,
\begin{itemize}
\itemsep0cm
\item $\mathcal{B}_k M \beta^{k+1}(1) = \mathcal{B}_{k+1}(1) = \mathcal{B}_{k+1}M\beta^0(1)$. 
\item $|\mathcal{B}_k(M\beta^i(1))| < |\mathcal{B}_k(M\beta^i(101))| < |\mathcal{B}_k(M\beta^{i+1}(1))|$, for all $i$ in $\{0, 1, \ldots, k\}$.
\end{itemize}
Therefore by Lemma~\mbox{\ref{lem:w beta M}}, $u_n$ is a strong bispecial factor of $\mathbf{w}_{\beta,M}$ if $n$ is even, and it is a weak bispecial factor of $\mathbf{w}_{\beta,M}$ if $n$ is odd. As, for any factor, $m(u) =1$ if it is strong bispecial, $m(u) = -1$ if it is weak bispecial and $m(u) = 0$ otherwise, by Formula~\mbox{\ref{eq:cassaigne}}, for all $n \geq 0$,
\begin{itemize}
\itemsep0cm
\item $s(|u_n|+1)-s(|u_n|) = 1$ if $n$ is even,
\item $s(|u_n|+1)-s(|u_n|) = -1$ if $n$ is odd,
\item $s(n+1)-s(n) = 0$ if $n \not \in \{ |u_m| \mid m \geq 0\}$.
\end{itemize}
Consequently, since $s(0) = p(1)-p(0)= 1$, for all $n \geq 0$, we have $p(n+1)-p(n) = s(n) \in \{1,2\}$.
\end{proof}

\subsection{Another (easier?) version of the $S$-adic conjecture}

A natural idea to try to understand the conjecture is to consider examples composed of well-known morphisms. For instance, one could consider the \textit{Fibonacci morphism} $\varphi = [01,0]$ whose fixed point is a Sturmian sequence and the \textit{Thue-Morse morphism} $\mu = [01,10]$ whose fixed points both have a sub-linear complexity. We have

\begin{proposition}
\label{proposition: fibo+TM}
If $S = \{\varphi, \mu\}$ where $\varphi$ and $\mu$ are defined above, any $S$-adic sequence is linearly recurrent.
\end{proposition}

\begin{proof}
Let $S = \{ \mu, \varphi\}$.
Let $\mathbf{w}$ be an $S$-adic sequence directed by $(\sigma_n)_{n \in \mathbb{N}}$ and, for all $k \in \mathbb{N}$, let $\mathbf{w}^{(k)}$ be the $S$-adic sequence directed by $(\sigma_n)_{n \geq k}$.

By definition of $\varphi$ and $\mu$, as  $\mathbf{w}^{(k)} \in \{\varphi(\mathbf{w}^{(k+1)}), \mu(\mathbf{w}^{(k+1)})\}$, word $111$ does not occur in $\mathbf{w}^{(k)}$ for all $k \geq 1$. 
Note also that  $000$ does not occur in $\mu(\{0,1\}^*)$, and $0000$ does not occur in $\varphi(\mathbf{w}^{(k+1)})$ since otherwise $111$ should occur in $\mathbf{w}^{(k+1)}$. Thus both words $0000$ and $111$ does not occur in $\mathbf{w}^{(k)}$ for all $k \geq 0$. This implies that the gap between two occurrences of $01$ and $10$ is at most 5 (7 is the maximal length of a shortest word on the form $10u10$ or on the form $01u01$).

Now observe that $11$ occurs in $\mathbf{w}^{(k)}$ only in the image of $\sigma_k(10)$ (and only if $\sigma_k = \mu$). This implies that the gap between two occurrences of $11$ is at most $12$.

Finally observe that $00$ occurs in $\mathbf{w}^{(k)}$ only in the image of $\sigma_1(10)$ or in the image of $\sigma_k(11)$ (when $\sigma_k = \varphi$). As factors of $\mathbf{w}^{(k)}$ that start with $10$ or $11$, end with $10$ or $11$ and do not contain any other occurrences of $10$ or $11$ have length at most $7$, the gap between two consecutive occurrences of $00$ is at most 12.

We have just proved that the length of the largest gap between two occurrences of a word of length $2$ in all words $\mathbf{w}^{(k)}$ is bounded. By choice of $S$, $\mathbf{w}$ is primitive $S$-adic. Hence by a result of Durand~\cite{Durand_corrigentum}, $\mathbf{w}$ is linearly recurrent.
\end{proof}

Let us say that a set of morphisms $S$ is \textit{good-adic} if all $S$-adic sequences have a sub-linear complexity.
Proposition~\mbox{\ref{proposition: fibo+TM}} provides an example of such a set. Many other examples are known as:
\begin{itemize}
	\item any singleton $\{ f \}$ with $f$ a morphism with fixed points of linear complexity;
	\item	the set of Sturmian morphisms $\{[01,1],[10,1],[0,01],[0,10]\}$;
	\item	the set of episturmian morphisms over an alphabet $A$: $S = \{ L_a, R_a \mid a \in A \}$ with for all $a \in A$, $L_a(a)=R_a(a) = a$, $L_a(b)=ab$ and $R_a(b)=ab$ for $b \neq a$;
	\item 	any finite set $S$ that contains only uniform morphisms (see Corollary~\ref{cor: durand uniform S-adic});
	\item 	any finite set $S$ that contains only strongly primitive morphisms (see Corollary~\ref{cor: durand 1 morphisme fortement primitif}).
\end{itemize}

Note that if $S$ is good-adic then, for any morphism $f$ in $S$ admitting an infinite fixed point, this fixed point must have a sub-linear complexity. But this necessary condition is certainly not the only one.

\begin{question}
What are good-adic sets of morphisms?
\end{question}

\subsection{About return words}
\label{subsection return words}

If $u$ is a factor of a sequence $\mathbf{w}$, a \textit{return word} to $u$ is a word $r$ such that $ru$ is a factor of $\mathbf{w}$ that admits $u$ as a prefix and that contains only two occurrences of $u$. Return words to $u$ in $\mathbf{w}$ actually correspond to the gaps between two occurrences of $u$ in $\mathbf{w}$.

Durand \mbox{\cite{Durand_UR}} proved that primitive {morphic} sequences, \textit{i.e.}, sequences defined as the image, under a morphism, of a fixed point of a primitive morphism can be characterised using return words. Hence, it is quite natural to ask whether such a result exists for $S$-adic sequences with sub-linear complexity.

There exist many examples of $S$-adic sequences $\mathbf{w}$ with sub-linear complexity for which there exists an integer $K$, such that all factors of $\mathbf{w}$ have at most $K$ return words.
For instance, Vuillon~\mbox{\cite{Vuillon_sturmian}} proved that Sturmian sequences are exactly sequences whose  factors have exactly two return words. Justin and Vuillon \mbox{\cite{Justin-Vuillon}} extended this result to the family of Arnoux-Rauzy sequences (also called strict episturmian sequences) whose factors have exactly $m$ return words with $m$ the cardinality of the alphabet. Balkov\'{a}, Pelantov\'{a} and Steiner~\mbox{\cite{Balkova-Pelantova-Steiner}} characterized sequences whose factors have all the same number of return words showing they have all a sub-linear complexity. Improving a result by Ferenczi~\mbox{\cite{Ferenczi}}, Leroy~\mbox{\cite{Leroy_these}} proved that sequences whose complexity verify ultimately $1 \leq p(n+1)-p(n) \leq 2$ are $S$-adic and {infinitely many} of their factors admit two or three return words.

Despite of previous example, next proposition show that there exist $S$-adic sequences with sub-linear complexity {such  that any long factor has many return words (``many'' depending on ``long'').} 

\begin{example}
\label{ex: plein de mots de retours}
Let us consider the three morphisms $p$, $\sigma$ and $s$ defined by
\begin{eqnarray*}
	p :	\begin{cases}
			a_1 \mapsto 0	\\
			a_2 \mapsto 1	\\
			b_1 \mapsto 1	\\
			b_2 \mapsto 0	
		\end{cases}
	\quad 
	\sigma :	\begin{cases}
					a_1 \mapsto a_1 a_2 a_1	\\
					a_2 \mapsto a_2 a_2 a_2	\\
					b_1 \mapsto b_1 b_2 b_1	\\
					b_2 \mapsto b_2 b_2 b_2	
				\end{cases}
	\quad
	s :	\begin{cases}
			0 \mapsto a_1	\\
			1 \mapsto b_1
		\end{cases}
\end{eqnarray*}
For all $n \geq 1$ we let $\pi_n$ denote the morphism $p \sigma^n s$; we have
\[
	\pi_n : \begin{cases}
				0 \mapsto 0 1^{\mathbf{e}_0} 0 1^{\mathbf{e}_1} 0 1^{\mathbf{e}_2} \cdots = \left( \prod_{i = 0}^{2^{n-1}-1}  0 1^{\mathbf{e}_i} \right) 0 	\\
				1 \mapsto 1 0^{\mathbf{e}_0} 1 0^{\mathbf{e}_1} 1 0^{\mathbf{e}_2} \cdots = \left( \prod_{i = 0}^{2^{n-1}-1}  1 0^{\mathbf{e}_i} \right) 1 	\\
			\end{cases},
\]
where $\mathbf{e}$ is the fixed point of the morphism ${\rm Exp}$ defined over the infinite alphabet $\{ 3^n \mid n \in \mathbb{N}\}$ by
\[
	\forall n \in \mathbb{N}, \quad  {\rm Exp}(3^n) = 1 3^{n+1}, 	
\]
\textit{i.e.},
\[
	\mathbf{e} = 1 3 1 3^2 1 3 1 3^3 1 3 1 3^2 1 3 1 3^4 1 3 1 3^2 1 3 1 3^3 1 3 1 3^2 1 3 \cdots.
\]
Now let us consider the sequence 
\[
	\mathbf{w}_{\pi} = \lim_{n \to + \infty} \pi_1 \pi_2 \cdots \pi_n (0^{\omega}).
\]

\begin{proposition}
The sequence $\mathbf{w}_{\pi}$ defined above is uniformly recurrent, has a sub-linear complexity and for all integers $k$, there is a length $\ell_k$ such that all factors of $\mathbf{w}_{\pi}$ of length at least $\ell_k$ have at least $k$ return words in $\mathbf{w}_{\pi}$.
\end{proposition}

\begin{proof}
The uniform recurrence is a direct consequence of Theorem~\ref{thm: characterization S-adicity minimal}. Since $\mathbf{w}$ is $S$-adic with $S = \{p,\sigma,s\}$, it is a consequence of Corollary~\ref{cor: durand uniform S-adic} that $\mathbf{w}_{\pi}$ has a sub-linear complexity.

Let us prove that all sufficiently long factors of $\mathbf{w}_{\pi}$ have many return words. For all $k \geq 1$, we let $\mathbf{w}^{(k)}$ denote the sequence
\[
	\mathbf{w}^{(k)} = \lim_{n \to \infty}	\pi_k \pi_{k+1} \cdots \pi_n (0^{\omega}).
\]
We obviously have $\mathbf{w}^{(1)} = \mathbf{w}_{\pi}$ and for all $k \geq 1$, $\mathbf{w}_{\pi} = \pi_1 \cdots \pi_k (\mathbf{w}^{(k+1)})$. We also have $|\pi_k(0)| = |\pi_k(1)| = 3^k$ for all $k$ and we let $\ell_k$ denote the length
\[
	|\pi_1 \pi_2 \cdots \pi_k(0)| = |\pi_1 \pi_2 \cdots \pi_k(1)| = 3^{\frac{k(k+1)}{2}}.
\]

Due to the shape of morphisms $\pi_1$, any factor $u$ of $\mathbf{w}_\pi$ of length at least equal to $24$ can be uniquely decomposed into images $\pi_1(0)$ and $\pi_1(1)$, \textit{i.e.}, there is a unique word $v$ in $\mathbf{w}^{(2)}$ such that $u \in {\rm Fac}(\pi_1(v))$ and such that if $u \in {\rm Fac}(\pi_1(v'))$, then $v \in {\rm Fac}(v')$. Indeed, given such a factor $u$, either it contains an occurrence of $00$ or of $11$ (which uniquely determine how to decompose the factor), or it is equal to $\pi_1(10101010)$ or to $\pi_1(01010101)$ (neither $101010101$, nor $010101010$ occur in $\mathbf{w}^{(2)}$).

Similarly, we see that, for all $k \geq 2$, any factor of $\mathbf{w}^{(k)}$ of length at least equal to $7$ can be uniquely decomposed into $\pi_k(0)$ and $\pi_k(1)$.

Let $n$ be a positive integer (suppose it is large). The sequence $(l_k)_{k \geq 1}$ is increasing so there is a unique positive integer $k$ such that $l_{k-1} \leq n < l_k$. Consequently, all factors of length $n$ of $\mathbf{w}_{\pi}$ belong to ${\rm Fac} \left( \pi_1 \cdots \pi_k \left( \{0,1\}^2 \right) \right)$. From what precedes, there is a unique word $v$ in ${\rm Fac}\left(\pi_k\left(\{0,1\}^2\right)\right) \subset {\rm Fac}\left(\mathbf{w}^{(k)}\right)$ such that $u \in {\rm Fac}(\pi_1 \cdots \pi_{k-1}(v))$ and any word $v'$ over $\{0,1\}$ such that $u \in {\rm Fac} \left( \pi_1 \cdots \pi_{k-1}(v') \right)$ contains $v$ as a factor. 

By unicity of $v$, the number of return words to $u$ in $\mathbf{w}_{\pi}$ is at least equal to the number of return words to $v$ in $\mathbf{w}^{(k)}$ (some return words could also occur in $\pi_1 \cdots \pi_{k-1}(v)$) so we only have to show that the number of return words to $v$ in $\mathbf{w}_k$ is at least linear in $k$. This will prove the result since $k$ increases with $n$.

First, if $|v| \geq 7$, we have seen that there is a unique word $x \in \{0,1,00,01,10,11\}$ such that $v \in {\rm Fac}(\pi_k(x))$ and such that if $v \in {\rm Fac}(\pi_k(y))$, then $x \in {\rm Fac}(y)$. The number of return words to $v$ in $\mathbf{w}^{(k)}$ is at least equal to the number of return words to $x$ in $\mathbf{w}^{(k+1)}$ and we let the reader check that it is at least linear in $k$.

Let us suppose that $|v|$ is less than 7 and that the factor $000$ occurs in $v$ (the case $111 \in {\rm Fac}(v)$ is similar). Then, from the shape of $\pi_k$, $v$ only occurs in $\mathbf{w}^{(k)}$ as factor of $\pi_k(1)$. The number of return words to $v$ in $\mathbf{w}^{(k)}$ is then at least equal to the number of return words to $1$ in $\mathbf{w}^{(k+1)}$ and we let the reader check that it is at least linear in $k$.

Now suppose that neither $000$ nor $111$ occur in $v$ (still with $|v| < 7$). Since the number of return words to a factor is the same as the number of return words to the smallest bispecial factor containing it, we only have to count the number of return words to bispecial factors of $\mathbf{w}^{(k)}$ that are smaller than $7$ and that do not contain $000$ neither $111$ as factors. We let the reader check that these factors are exactly
\[
	\left\{	0,1,00,01,10,11,010,101,0101,1010,01010,01011,11010,10101,010101,101010	\right\}.
\]
The bispecial factors $01010$, $10101$, $010101$ and $101010$ only occur in $\mathbf{w}^{(k)}$ as factors of $\pi_k(01)$ and $\pi_k(01)$ but are not factor of $\pi_k(0)$ neither of $\pi_k(1)$. Since the number of words filling the gaps between occurrences of $01$ and $10$ in $\mathbf{w}^{(k+1)}$ is linear in $k$, these bispecial factors have a number of return words which is linear in $k$ too.

We let the reader check that the other bispecial factors have a number of return words in $\mathbf{w}^{(k)}$ which is linear in $k$ (each return word containing a different highest power of $0$ or of $1$ depending on the bispecial factor). This completes the proof.
\end{proof}
\end{example}

To end with return words, let us observe that it is possible to build sequences whose complexity is not sub-linear and in which infinitely many factors have a bounded number of return words. Indeed, for instance, the sequence $\gamma^{\omega}(0)$ (see Example~\mbox{\ref{ex: morse+n^2}}) has a quadratic complexity and any factor $\gamma^n(1)$ admits two return words.
After last observation it seems difficult to characterize sequences with a sub-linear complexity using return words. Nevertheless next interesting question is open.

\begin{question}
Let $\mathbf{w}$ be a sequence such that, for some integer $K$, all factors of $\mathbf{w}$ has at most $K$ return words. Is it true that $\mathbf{w}$ is $S$-adic for some suitable {set} $S$ and has linear complexity?
\end{question}

\subsection{What about the number of distinct powers?}
\label{subsection: distinct powers}

As already mentioned, it is known that the factor complexity of $k$-power-free morphic purely sequences grows at most like $n \log n$. In this section, we explore links between bounding the number of distinct exponents of factors instead of the maximal exponent of factors and the factor complexity. 

Let us introduce the following notation: given a language $L$ and a word $u = u_1 \cdots u_{|u|} \in L$, ${\rm Pow}(u, L)$ is the set of all non-negative integers $i$ such that there are some words $p$ and $s$ such that $pu^i s$ belongs to $L$, 
$|p| \leq |u|$, $|s|\leq |u|$, $p$ is not a suffix of $u$ and $s$ is not a prefix of $u$.
With the notation of previous section, we have 
${\rm Pow}(\pi_1\cdots \pi_n(0), {\rm Fac}(\mathbf{w}_\pi)) = \{3^1, 3^2, \ldots, 3^n\}$, which shows the existence of an $S$-adic sequence with an at most linear complexity that have factors with an unbounded number of distinct exponents.

The previous phenomenon also holds for $S$-adic sequences that do not have a sub-linear complexity.
Indeed this is the case of the $S$-adic sequence considered in Proposition~\ref{prop: morse+n^2} when the sequence $(k_n)_{n \in \mathbb{N}}$ is unbounded, as one can observe that
\[
	{\rm Pow}(\gamma^{k_0} \mu \gamma^{k_1} \mu \gamma^{k_2} \mu \cdots \gamma^{k_n} \mu(1), {\rm Fac}(\mathbf{w}_{\gamma,\mu})) = \{1, 2, \ldots, k_{n+1}+2\}.
\]

If it seems difficult from previous discussion to get an $S$-adic characterization of sequences with sub-linear complexity using the number of distinct powers of their factors, next proposition inspired by word $\mathbf{w}_{\gamma,\mu}$ shows that the number of distinct powers can be used to prove that a sequence does not have a sub-linear complexity. Recall that a word $u$ is \textit{primitive} if it is not a power of a smaller word $v$, \textit{i.e.}, there is no integer $k \geq 2$ such that $u = v^k$.

\begin{proposition}
\label{prop: puissances}
Let $\mathbf{w}$ be a recurrent sequence over $A$. If there is a constant $C > 1$, a sequence $(u_n)_{n \in \mathbb{N}}$ of primitive
words in ${\rm Fac}(\mathbf{w})$ and an increasing sequence $(k_n)_{n \in \mathbb{N}}$ of positive integers such that $(|u_n|)_{n \in \mathbb{N}}$ is increasing and such that for all $n$, all integers $i$ with $\frac{k_n}{C} \leq i \leq k_n$ belong to ${\rm Pow}(u_n,{\rm Fac}(\mathbf{w}))$, then $\mathbf{w}$ does not have a sub-linear complexity.
\end{proposition}

Before proving the result, let us recall the classical following result (see for instance Proposition 3.1.2 in \cite{Lothaire1983}).

\begin{proposition}
\label{prop: xy yx}
If $x$ and $y$ are two non-empty words such that $xy = yx$, then there is a word $v$ smaller than $x$ and $y$ such that $x = v^n$ and $y = v^m$ for some positive integers $n$ and $m$.
\end{proposition}

\begin{proof}
Let us give a lower bound on the number of factors of length $k_n |u_n|$ of $\mathbf{w}$. When $C$ is at least equal to $2$, the property is also true with $2 - \varepsilon$ instead of $C$ for all $\varepsilon > 0$, thus, without loss of generality, we assume $C < 2$. 

Let $n$ be a positive integer and let $i$ in $\{0,1,\dots,k_n-\left\lceil\frac{k_n}{C}\right\rceil\}$.
By hypothesis the word $u_n^{k_n-i}$ is a factor of $\mathbf{w}$.

Let us study the return words to $u_n^{k_n-i}$. Since $u_n$ is a primitive word, a return word $r$ to $u_n^{k_n-i}$ is either $u_n$ or has length $|r| > (k_n-i)|u_n|$. Indeed, if not, one of the following holds true:
\begin{enumerate}
	\item $|r| = j|u_n|$ with $2 \leq j \leq k_n-i$; in this case, we have $r = u_n^j$ which is not a return word to $u_n$;
	\item $|r| < (k_n-i)|u_n|$ and $|r|$ is not a multiple of $|u_n|$; in this case, there is a non-empty word $u'$ smaller than $u_n$ such that $u' u_n = u_n u'$ and, by Proposition~\ref{prop: xy yx}, $u_n$ is not a primitive word.
\end{enumerate}
 
Now, since $k_n-i$ belongs to ${\rm Pow}(u_n,{\rm Fac}(\mathbf{w}))$ and $\mathbf{w}$ is recurrent, there are two words $p_i$ and $s_i$ such that $p_i u_n^{k_n-i} s_i$ belongs to ${\rm Fac}(\mathbf{w})$, $|p_i| \leq |u_n|$, $|s_i| \leq |u_n|$ and $u_n$ does not admit $p_i$ as a suffix neither $s_i$ as a prefix.
Thus, from what precedes, $p_i$ is a suffix of a return word $r_i$ to $u_n^{k_n-i}$ with $|r_i| > (k_n-i)|u_n|$ and such that $r_i u_n^{k_n-i} s_i$ is a factor of the recurrent word $\mathbf{w}$.

By hypothesis, we have $C < 2$ and $k_n-i  \geq \left\lceil \frac{k_n}{C} \right\rceil$ so it comes
\[
	|r_i u_n^{k_n-i}| > 2 \left\lceil \frac{k_n}{C} \right\rceil |u_n| > k_n |u_n|.
\]

Let $v_i$ denote the suffix of length $k_n|u_n|$ of $r_i u_n^{k_n-i}$ and, when $i \geq 1$, let $x_i$ denote a factor of $\mathbf{w}$ of length $(k_n+i)|u_n|$ that admits $v_i s_i$ as a prefix. Let also $x_0 = u_n^{k_n}$). For a given word $w = w_1 \cdots w_{|w|}$, we let $w[r:s]$ denote its factor $w_r \cdots w_{s-1}$. By definition of return words, the word $u_n^{k_n-i}$ occurs only once in $v_i$. Thus all the words 
\[
	x_i[\ell:\ell+k_n|u_n|], \quad \ell \in \{1,2,\dots,i|u_n|\}, 
\]
are different.

Let us show that if $i$ and $j$ belong to $\{1,\dots,k_n-\left\lceil\frac{k_n}{C}\right\rceil\}$, $i < j$, then all words 
\begin{eqnarray*}
	x_i[\ell : \ell + k_n |u_n|], 		& \ell \in \{1,2,\dots,i|u_n|\} \\
	x_j[m : m + k_n |u_n|], 			& m \in \{1,2,\dots,j|u_n|\}
\end{eqnarray*}
are also distinct (by construction they are not equal to $x_0$). Suppose by contrary that there are some integers $\ell \in \{1,2,\dots,i|u_n|\}$ and $m \in \{1,2,\dots,j|u_n|\}$ such that
\[
	x_i[\ell:\ell+k_n|u_n|] = x_j[m:m+k_n|u_n|] = X. 
\]
Since $\ell \leq i|u_n|$ and $m \leq j|u_n|$, there are some words $\alpha$, $\beta$, $\gamma$ and $\delta$ such that 
\[
	X = \alpha u_n^{k_n-i} \beta = \gamma u_n^{k_n-j} \delta.
\]
We have $|\alpha|>|\gamma|$ otherwise $u_n^{k_n-j}$ would occur twice in $v_j$. 
Let us show that we also have $|\alpha| < |\gamma| + (k_n-j)|u_n|-|u_n|$. As $k_n - i > k_n - j \geq \left\lceil \frac{k_n}{C} \right\rceil > \frac{k_n}{2}$, we have
\[
	(k_n-i)|u_n| \geq (k_n-j+1)|u_n| > \frac{k_n|u_n|}{2}+|u_n|.
\]
Thus, if $|\alpha| \geq |\gamma| + (k_n-j)|u_n|-|u_n|$, we have 
\begin{eqnarray*}
	|\alpha| + (k_n-i)|u_n|  &\geq	& |\gamma| + (k_n-j)|u_n|-|u_n| + (k_n-j+1)|u_n|	\\
						& >		& |\gamma| + k_n |u_n|,
\end{eqnarray*}
which is a contradiction because $|\alpha u^{k_n-i}\beta|=k_n|u_n|$.

We have just proved that 
\[
	|\gamma| < |\alpha| < |\gamma| + (k_n-j)|u_n|-|u_n| = |\gamma u_n^{k_n-j}| - |u_n|.
\]
This implies that, if $p$ is the prefix of length $(|\alpha|-|\gamma|) \mod |u_n|$ of $u_n$, we have $u_np = pu_n$. Since $|\alpha|-|\gamma| \neq 0 \mod |u_n|$ (otherwise $s_j$ would be a prefix of $u_n$), Proposition~\ref{prop: xy yx} implies that $u_n$ is not primitive, hence a contradiction.  

Now we can conclude the proof: since all words $x_i[\ell:\ell+k_n|u_n|]$, $i \in \{1,\dots,k_n-\left\lceil\frac{k_n}{C}\right\rceil\}$, $\ell \in \{1,2,\dots,i|u_n|\}$, are distinct and different from $x_0$, we have
\begin{eqnarray*}
	p_{\mathbf{w}}(k_n|u_n|) 	&	\geq	& 1	+ \sum_{i=1}^{k_n - \left\lceil\frac{k_n}{C}\right\rceil} i |u_n|	\\
				&	\geq 	&	|u_n|	\sum_{i=1}^{\frac{C-1}{C}k_n-1} i	\\
				&	\geq 	&	\frac{|u_n|}{2}	 \left(\frac{C-1}{C}k_n \left(\frac{C-1}{C}k_n-1 \right) \right)
\end{eqnarray*}
and we deduce that the complexity is not sub-linear.
\end{proof}

\section{Beyond linearity}

Until now, we have provided several examples showing that various natural approaches to characterize $S$-adic sequences that have a sub-linear complexity does not appear to be promising. To conclude this paper, we raise a new problem related to $S$-adicity and, more precisely, to everywhere growing $S$-adic sequences.

For purely morphic sequences, the complexity function can have only 5 asymptotic behaviours and only depends on the growth rate of images (see Theorem~\ref{thm: pansiot}). For $S$-adic sequences we have seen in Section~\ref{subsection S-adic} that things are highly more complicated. However, it has been proved in~\cite{Ferenczi} (see also~\cite{Leroy,Leroy-Richomme}) that any uniformly recurrent sequence with sub-linear complexity is everywhere growing $S$-adic with ${\rm Card}(S)<+\infty$. This is a kind of generalization of the third point of Theorem~\ref{thm: pansiot}. Moreover, one can check that all examples considered in previous sections (and more generally all $S$-adic representations of well-known families of sequences such as codings of rotations, codings of interval exchanges, etc.) are everywhere growing. It is also interesting to note that for purely morphic sequences, the class of highest complexity $\Theta(n^2)$ can be reached only by morphisms with bounded letters (still Theorem~\ref{thm: pansiot}). Furthermore, up to now, Cassaigne's constructions (Proposition~\ref{prop:Cassaigne adique}) are the only ones that allow to build $S$-adic sequences with arbitrarily high complexity and they admit several bounded letters. Consequently, the fact that the length of all images tends to infinity with $n$ seems to be important to get a \textit{reasonably low complexity}. 

Observe that, given an $S$-adic sequence $\mathbf{w}$, the everywhere growing property is not a necessary condition for $\mathbf{w}$ to have a low complexity. Indeed, Cassaigne's constructions also hold for sequences with low complexity. One can also think to the \textit{Chacon substitution} $\varrho$ defined by $\varrho(0)=0010$ and $\varrho(1)=1$ whose fixed point $\varrho^{\omega}(0)$ has complexity $p(n)=2n+1$ for all $n$ (see~\cite{chacon}). However, \textit{the existence} of an everywhere growing $S'$-adic representation of $\mathbf{w}$ could be necessary (this is the case for the Chacon substitution and for any uniformly recurrent sequence with sub-linear complexity). 

That property is neither a sufficient condition since the sequence $\mathbf{w}_{\gamma,E}$ of Example~\ref{ex: morse+n^2} satisfies it and does not always have a sub-linear complexity. However the following question seems to be natural.

\begin{question}
Is it possible to reach any high complexity with everywhere growing $S$-adic sequences? 
\end{question}

This question seems to be a new non-trivial problem. Proposition~\ref{prop: S-adic bounded by n log n} below provides a partial answer to that question. Indeed, it deals with \textit{expansive} $S$-adic sequences, \textit{i.e.}, with $S$-adic sequences such that for all morphisms $\sigma$ in $S$ and all letters $a$, we have $|\sigma(a)| \geq 2$. The proof can be found in~\cite{Leroy} and involves techniques similar to those used in~\cite{D0L} for D0L systems.

\begin{proposition}
\label{prop: S-adic bounded by n log n}
If $\mathbf{w}$ is an expansive $S$-adic sequence with ${\rm Card}(S)< +\infty$, then $p_{\mathbf{w}}(n) \in O(n \log n)$.
\end{proposition}

Example~\ref{example: complexity n log n} shows that this bound is the best one we can obtain.

\begin{example}
\label{example: complexity n log n}
Let $\beta$ be the morphism
\[
	\vartheta :	\begin{cases}
					0 \mapsto 0120	\\
					1 \mapsto 11	\\
					2 \mapsto 222
				\end{cases}
\]
and consider its fixed point $\mathbf{w}=\vartheta^{\omega}(a)$. It can be seen as an expansive $\{\vartheta\}$-adic sequence and we know from Theorem~\ref{thm: pansiot} that $p_{\mathbf{w}}(n) =\Theta(n \log n)$.
\end{example}

\def\cprime{$'$}

\bigskip
\hrule
\bigskip

\noindent 2010 {\it Mathematics Subject Classification}:
Primary 68R15; Secondary 37B10.

\noindent \emph{Keywords: } 
Factor complexity; $S$-adicity; morphisms.

%\bigskip
%\hrule
%\bigskip
%
%\vspace*{+.1in}
%\noindent
%Received ;
%revised versions received  .
%Published in {\it Journal of Integer Sequences},.
%
%\bigskip
%\hrule
%\bigskip
%
%\noindent
%Return to
%\htmladdnormallink{Journal of Integer Sequences home page}{http://www.cs.uwaterloo.ca/journals/JIS/}.
%\vskip .1in


\begin{thebibliography}{ELR75}

\bibitem[AB07]{Adamczewski}
B.~Adamczewski and Y.~Bugeaud,
\newblock On the complexity of algebraic numbers {I}: {E}xpansions in integer
  bases,
\newblock {\em Ann. of Math. (2)} 165 (2007), 547--565.

\bibitem[Abe03]{Aberkane}
A.~Aberkane,
\newblock Words whose complexity satisfies {$\lim\frac{p(n)}{n}=1$},
\newblock {\em Theoret. Comput. Sci.} 307 (2003), 31--46.

\bibitem[All94]{allouche_survey}
J.-P. Allouche,
\newblock Sur la complexit\'e des suites infinies,
\newblock {\em Bull. Belg. Math. Soc. Simon Stevin} 1 (1994), 133--143.

\bibitem[AR91]{Arnoux-Rauzy}
P.~Arnoux and G.~Rauzy,
\newblock Repr\'esentation g\'eom\'etrique de suites de complexit\'e {$2n+1$},
\newblock {\em Bull. Soc. Math. France} 119 (1991), 199--215.

\bibitem[Ber80]{Berstel1980}
J.~Berstel,
\newblock Mots sans carr\'e et morphismes it\'er\'es,
\newblock {\em Discrete Math.} 29 (1980), 235--244.

\bibitem[BPS08]{Balkova-Pelantova-Steiner}
L.~Balkov{\'a}, E.~Pelantov{\'a}, and W~Steiner,
\newblock Sequences with constant number of return words,
\newblock {\em Monatsh. Math.} 155 (2008), 251--263.

\bibitem[BR10]{CANT}
V.~Berth\'{e} and M.~Rigo editors,
\newblock {\em Combinatorics, Automata and Number Theory}, volume 135 of {\em
  Encyclopedia of Mathematics and its Applications},
\newblock Cambridge University Press, Cambridge, 2010.

\bibitem[Cas96]{Cassaigne_big_thm}
J.~Cassaigne,
\newblock Special factors of sequences with linear subword complexity,
\newblock In {\em Developments in Language Theory, {II} ({M}agdeburg, 1995)},
  World Sci. Publ., River Edge, NJ, 1996, pp. 25--34.

\bibitem[Cas97]{Cassaigne_resume}
J.~Cassaigne,
\newblock Complexit\'e et facteurs sp\'eciaux,
\newblock {\em Bull. Belg. Math. Soc. Simon Stevin} 4 (1997), 67--88.

\bibitem[Cas03]{Cassaigne_intermediate}
J.~Cassaigne,
\newblock Constructing infinite words of intermediate complexity,
\newblock in {\em Developments in Language Theory}, {\em Lect.
  Notes in Comput. Sci.}, Vol. 2450, Springer, Berlin, 2003, pp. 173--184.

\bibitem[CN03]{Cassaigne-Nicolas}
J.~Cassaigne and F.~Nicolas,
\newblock Quelques propri\'et\'es des mots substitutifs,
\newblock {\em Bull. Belg. Math. Soc. Simon Stevin} 10 (2003), 661--676.

\bibitem[Cob68]{Cobham_hartmanis}
A.~Cobham,
\newblock On the {H}artmanis-{S}tearns problem for a class of tag machines,
\newblock in {\em Proceedings of the 9th Annual Symposium on Switching and
  Automata Theory (swat 1968)}, Washington, DC, USA, 1968. IEEE
  Computer Society, pp. 51--60.

\bibitem[Dev08]{Deviatov}
R.~Deviatov,
\newblock On subword complexity of morphic sequences,
\newblock in {\em Computer science---theory and applications}, 
  {\em Lect. Notes in Comput. Sci.}, Vol. 5010, Springer, Berlin, 2008, pp. 146--157. 

\bibitem[DL06]{Damanik-Lenz}
D.~Damanik and D.~Lenz,
\newblock Substitution dynamical systems: characterization of linear
  repetitivity and applications,
\newblock {\em J. Math. Anal. Appl.} 321 (2006), 766--780.

\bibitem[Dur]{Durand_preprint}
F.~Durand,
\newblock Decidability of uniform recurrence of morphic sequences,
\newblock submitted.

\bibitem[Dur98]{Durand_UR}
F.~Durand, 
\newblock A characterization of substitutive sequences using return words,
\newblock {\em Discrete Math.} 179 (1998), 89--101.

\bibitem[Dur00]{Durand_LR}
F.~Durand,
\newblock Linearly recurrent subshifts have a finite number of non-periodic
  subshift factors,
\newblock {\em Ergodic Theory Dynam. Systems} 20 (2000), 1061--1078.

\bibitem[Dur03]{Durand_corrigentum}
F.~Durand,
\newblock Corrigendum and addendum to: ``{L}inearly recurrent subshifts have a
  finite number of non-periodic subshift factors'',
\newblock {\em Ergodic Theory Dynam. Systems} 23 (2003), 663--669.

\bibitem[ELR75]{D0L}
A.~Ehrenfeucht, K.~P. Lee, and G.~Rozenberg,
\newblock Subword complexities of various classes of deterministic
  developmental languages without interactions,
\newblock {\em Theoret. Comput. Sci.} 1 (1975), 59--75.

\bibitem[ER83]{D0L_m-free}
A.~Ehrenfeucht and G.~Rozenberg,
\newblock On the subword complexity of {$m$}-free {D0L} languages,
\newblock {\em Inform. Process. Lett.} 17 (1983), 121--124.

\bibitem[Fer95]{chacon}
S.~Ferenczi,
\newblock Les transformations de {C}hacon: combinatoire, structure
  g\'eom\'etrique, lien avec les syst\`emes de complexit\'e {$2n+1$},
\newblock {\em Bull. Soc. Math. France} 123 (1995), 271--292.

\bibitem[Fer96]{Ferenczi}
S.~Ferenczi,
\newblock Rank and symbolic complexity,
\newblock {\em Ergodic Theory Dynam. Systems} 16 (1996), 663--682.

\bibitem[Fer99]{ferenczi_survey}
S.~Ferenczi,
\newblock Complexity of sequences and dynamical systems,
\newblock {\em Discrete Math.} 206 (1999), 145--154.

\bibitem[Fog02]{Pytheas-Fogg}
N.~Pytheas Fogg,
\newblock {\em Substitutions in dynamics, arithmetics and combinatorics},
  volume 1794 of {\em Lecture Notes in Mathematics},
\newblock Springer-Verlag, Berlin, 2002.
\newblock Edited by V. Berth{\'e}, S. Ferenczi, C. Mauduit and A. Siegel.

\bibitem[Fog11]{Cassaigne_S-adic}
N.~Pytheas Fogg,
\newblock Terminologie {$S$}-adique et propri\'et\'es,
\newblock
  \url{https://www2.lirmm.fr/$\sim$monteil/hebergement/pytheas-fogg/terminologie\_s
\_adique.pdf}, 2011.

\bibitem[GJ09]{Glen-Justin}
A.~Glen and J.~Justin,
\newblock Episturmian words: a survey,
\newblock {\em Theor. Inform. Appl.} 43 (2009), 403--442.

\bibitem[Gri73]{Grillenberger}
C.~Grillenberger,
\newblock Constructions of strictly ergodic systems. {I}. {G}iven entropy,
\newblock {\em Z. Wahrscheinlichkeitstheorie und Verw. Gebiete} 25 (1972/73), 323--334.

\bibitem[JV00]{Justin-Vuillon}
J.~Justin and L.~Vuillon,
\newblock Return words in {S}turmian and episturmian words,
\newblock {\em RAIRO Theor. Inform. Appl.} 34 (2000), 343--356.

\bibitem[Kos98]{koskas}
M.~Koskas,
\newblock Complexit\'es de suites de {T}oeplitz,
\newblock {\em Discrete Math.}, 183 (1998), 161--183.

\bibitem[Ler12a]{Leroy_these}
J.~Leroy,
\newblock {\em Contribution to the resolution of the {$S$}-adic conjecture},
\newblock PhD thesis, Universit\'e de Picardie Jules Verne, 2012.

\bibitem[Ler12b]{Leroy}
J.~Leroy,
\newblock Some improvements of the {$S$}-adic conjecture,
\newblock {\em Adv. in Appl. Math.} 48(2012), 79 -- 98.

\bibitem[Lot97]{Lothaire1983}
M.~Lothaire,
\newblock {\em Combinatorics on words},
\newblock Cambridge Mathematical Library, Cambridge University Press,
  Cambridge, 1997.
\newblock Corrected reprint of the 1983 original.

\bibitem[Lot02]{Lothaire}
M.~Lothaire,
\newblock {\em Algebraic combinatorics on words}, volume~90 of {\em
  Encyclopedia of Mathematics and its Applications},
\newblock Cambridge University Press, Cambridge, 2002.

\bibitem[LR]{Leroy-Richomme}
J.~Leroy and G.~Richomme,
\newblock A combinatorial proof of {$S$}-adicity for sequences with sub-affine
  complexity, 
\newblock preprint.

\bibitem[MH40]{Morse-Hedlund}
M.~Morse and G.~A. Hedlund,
\newblock Symbolic dynamics {II}. {S}turmian trajectories,
\newblock {\em Amer. J. Math.} 62 (1940), 1--42.

\bibitem[MM10]{Mauduit-Moreira}
C.~Mauduit and C.~G. Moreira,
\newblock Complexity of infinite sequences with zero entropy,
\newblock {\em Acta Arith.} 142 (2010), 331--346.

\bibitem[NP09]{Nicolas-Pritykin}
F.~Nicolas and Y.~Pritykin,
\newblock On uniformly recurrent morphic sequences,
\newblock {\em Internat. J. Found. Comput. Sci.} 20 (2009), 919--940.

\bibitem[Pan83]{Pansiot_hierarchical}
J.-J. Pansiot,
\newblock Hi\'erarchie et fermeture de certaines classes de tag-syst\`emes,
\newblock {\em Acta Inform.} 20 (1983), 179--196.

\bibitem[Pan84]{Pansiot}
J.-J. Pansiot,
\newblock Complexit\'e des facteurs des mots infinis engendr\'es par morphismes
  it\'er\'es,
\newblock in {\em Automata, languages and programming ({A}ntwerp, 1984)},
  {\em Lect. Notes in Comput. Sci.}, Vol. 172, Springer,
  Berlin, 1984, pp. 380--389. 

\bibitem[Pan85]{Pansiot_subword}
J.-J. Pansiot,
\newblock Subword complexities and iteration,
\newblock {\em Bull. Eur. Assoc. Theor. Comput. Sci. EATCS} 26 (1985), 55--62.

\bibitem[Rot94]{Rote}
G.~Rote,
\newblock Sequences with subword complexity {$2n$},
\newblock {\em J. Number Theory} 46 (1994), 196--213.

\bibitem[RS80]{L_system}
G.~Rozenberg and A.~Salomaa,
\newblock {\em The mathematical theory of {L} systems}, volume~90 of {\em Pure
  and Applied Mathematics},
\newblock Academic Press Inc. [Harcourt Brace Jovanovich Publishers], New York,
  1980.

\bibitem[S{\'e}{\'e}85]{Seebold1985}
Patrice S{\'e}{\'e}bold,
\newblock Sequences generated by infinitely iterated morphisms,
\newblock {\em Discrete Appl. Math.} 11 (1985), 255 -- 264.

\bibitem[Vui01]{Vuillon_sturmian}
L.~Vuillon,
\newblock A characterization of {S}turmian words by return words,
\newblock {\em European J. Combin.} 22 (2001), 263--275.

\end{thebibliography}
\end{document}